\documentclass[pra,reprint,amsmath,amssymb,aps]{revtex4-2}

\usepackage[T1]{fontenc}
\usepackage[utf8]{inputenc}
\usepackage{graphicx}
\usepackage{dcolumn}
\usepackage{bm}
\usepackage{xcolor}
\usepackage{amsthm}
\usepackage{amsmath}
\usepackage{amssymb}
\usepackage{mathtools}
\usepackage{physics}
\usepackage[hidelinks]{hyperref}

\newtheorem{lemma}{Lemma}
\newtheorem{theorem}{Theorem}

\newtheorem{remark}{Remark}
\newtheorem{definition}{Definition}

\newcommand{\id}{\mathbb{I}}
\newcommand{\cH}{\mathcal{H}}
\newcommand{\cX}{\mathcal{X}}
\newcommand{\cY}{\mathcal{Y}}

\newcommand{\Dens}{\mathcal{D}}

\begin{document}

\title{Quantum work beyond classical (commuting) limits}
\author{Sumit Rout}
\thanks{These authors contributed equally to this work.}
\affiliation{International Centre for Theory of Quantum Technologies (ICTQT), University of Gda{\'n}sk, Jana Ba{\.z}ynskiego 8, 80-309 Gda{\'n}sk, Poland}
\author{Aravinth Balaji Ravichandran}
\thanks{These authors contributed equally to this work.}
\affiliation{International Centre for Theory of Quantum Technologies (ICTQT), University of Gda{\'n}sk, Jana Ba{\.z}ynskiego 8, 80-309 Gda{\'n}sk, Poland}
\author{Pawe{\l} Horodecki}
\affiliation{International Centre for Theory of Quantum Technologies (ICTQT), University of Gda{\'n}sk, Jana Ba{\.z}ynskiego 8, 80-309 Gda{\'n}sk, Poland}
\affiliation{Faculty of Applied Physics and Mathematics, Gda{\'n}sk University of Technology, Gabriela Narutowicza 11/12, 80-233 Gda{\'n}sk, Poland}
\author{Anubhav Chaturvedi}
\email{anubhav.chaturvedi@ug.edu.pl}
\affiliation{Division of Quantum Optics and Information, Institute of Theoretical Physics and Astrophysics, Faculty of Mathematics, Physics and Informatics, University of Gda\'nsk, 80-308 Gda\'nsk, Poland}

\begin{abstract}
Free energy fixes the maximum work of a thermodynamic process once the state and Hamiltonian are specified. 
A work-extraction task asks a different question: how much average work can a single device realize across several preparations and Hamiltonian settings? 
A classical work device is one whose Hamiltonian settings are mutually commuting. 
We place every branch at its best free-energy-limited work envelope and derive the corresponding classical limit on the task average. 
For pure preparations, the source is specified only by pairwise maximal-energy constraints: for each pair, the intrinsic maximal average energy under one common normalized Hamiltonian is bounded as part of the task data, while the work device is otherwise microscopically unrestricted. 
The benchmark is optimized over arbitrary-dimensional classical implementations. 
Incompatible Hamiltonian settings exceed this limit, even though every branch remains bounded by its own free-energy maximum. 
The advantage therefore does not arise in any single process, but in the average work of the task: incompatible Hamiltonians realize a value that no classical work device can attain. 
Hamiltonian incompatibility is thus a thermodynamic resource for work extraction.
\end{abstract}

\maketitle

\section{Introduction}
Thermodynamics sets the ultimate limits on the conversion of energy into useful work. 
At the microscopic scale, work extraction is governed by quantum thermodynamics: for a single process in contact with a bath at temperature \(T\)\cite{BreuerPetruccione2002}, once the input state and Hamiltonian are fixed, the maximum average extractable work is determined by the appropriate nonequilibrium free-energy difference \cite{Alicki1979,Spohn1978,HorodeckiOppenheim2013,SkrzypczykShortPopescu2014,Aberg2013,BrandaoHorodeckiNgOppenheimWehner2015,RichensMasanes2016,SafranekRosaBinder2023,WatanabeTakagi2026}. 
Thus, the work extractable from any single state-Hamiltonian instance is already fixed. 
A work-extraction task, however, does not concern a single process. 
It supplies several preparations to the same work device and probes several Hamiltonian settings. 
A task therefore requires a different limit: not the work extractable from one such instance, but the maximum average work that a single device can realize across the entire task when its Hamiltonian settings are required to commute. 
That is the missing classical limit of average work.

Existing quantum-thermodynamic work-extraction results close the single-process problem through free energy. 
Recent thermodynamic advantage protocols show that quantum resources can enhance work extraction or cooling in specific multi-setting tasks, including protocols based on steering, coherence, entanglement, and general resource-theoretic structure \cite{BiswasDattaGarciaPintos2025,HsiehGessner2024General,GarciaPintosBiswasDatta2026Robustness,BiswasDattaGarciaPintosCooling2025}. 
These works establish important quantum improvements, but their benchmarks are tied to specified resources, Hamiltonian architectures, or concrete work-extraction protocols. 
They do not give the device-level law needed here: the largest average work attainable by one arbitrary-dimensional work device when the only classical condition is mutual commutativity of the Hamiltonian settings.

This paper derives that law. 
The comparison is deliberately generous to classical devices. 
Branch by branch, every physical protocol is bounded by the corresponding free-energy maximum, and an imperfect protocol can only extract less \cite{KurizkiKofman2021}. 
We therefore evaluate the task at the best-case thermodynamic envelope. 
After this branch-wise envelope is fixed, the only remaining distinction is whether the Hamiltonian settings of one device commute. 
The resulting benchmark is an absolute classical average-work limit, not a bound on a particular classical implementation strategy.

On the device side, classicality is expressed entirely through the mutual commutativity of the Hamiltonian settings implemented across the task. 
No microscopic model of the work device is imposed. 
Its underlying Hilbert-space dimension, internal degrees of freedom, and work-extraction mechanism remain unrestricted. 
The classical comparator is therefore the full class of bounded Hamiltonian implementations satisfying this commutativity condition.

The task concerns the average over all branches. 
In each run it samples a preparation and a Hamiltonian setting; over repeated runs it estimates the prior-weighted task value. 
At the branch-wise free-energy envelope, entropy terms and partition-function terms are branch-local offsets fixed by the branch data. 
They do not encode whether different Hamiltonian settings of one device commute. 
The commutativity-sensitive part of the comparison is the internal-energy table \(\Tr(\rho_xH_y)\), and the classical law below is the optimal commuting bound on its task average under the same thermodynamic reference data.

A nontrivial device-level work law also requires thermodynamic source data. 
If the supplied preparations could encode an unconstrained branch label, the task average would impose no genuine limitation on one device. 
If the source energy scale were not fixed, a larger task average could simply reflect more supplied energy. 
For a thermodynamic law, the source restriction must therefore be expressed as normalized Hamiltonian-energy data rather than as an external dimension assumption or a postulated overlap constraint. 
We use pairwise maximal-energy constraints. 
For each pair of pure preparations, the task specifies constraints on the intrinsic quantity
\[
\eta_{xx'}
=
\max_{h\ge0,\operatorname{Tr}[h]=1}
\frac12\operatorname{Tr}\!\left[(\rho_x+\rho_{x'})h\right],
\]
where the same normalized Hamiltonian \(h\) is used for the two preparations in that pair. 
The task may impose either upper or lower constraints on this intrinsic pair quantity. 
For pure preparations, \(\eta_{xx'}=(1+|\langle\psi_x|\psi_{x'}\rangle|)/2\), so these thermodynamic constraints determine the pairwise transition amplitudes entering the task. 
Thus the source geometry is derived from normalized energy data rather than assumed microscopically. 
The right-hand side of a classical work law must therefore contain thermodynamic reference quantities: without them, a larger task average could reflect a different source energy scale rather than a failure of classical Hamiltonian compatibility. 
The operational structure of the task is shown in Fig.~\ref{fig:work-extraction-task}.
 
The first result is an exact classical average-work law for the three-preparation task. 
For three source preparations and two Hamiltonian settings, we derive the universal constraint obeyed by every arbitrary-dimensional device whose two Hamiltonian settings commute. 
Optimizing this law gives the absolute classical benchmark for the task. 
The second result shows that incompatible Hamiltonians exceed this benchmark, and that the symmetric source point and balanced reference averages are selected by the optimization rather than assumed. 
Finally, for balanced source families with many Hamiltonian settings, the classical benchmark becomes a simultaneous-alignment problem; this yields exact limiting robustness thresholds for noisy incompatible settings.

The principle generates a hierarchy of work-extraction tasks. 
The minimal task already separates incompatible Hamiltonians from every classical implementation. 
Larger thermodynamic source families strengthen the separation by turning the classical benchmark into a simultaneous-alignment problem. 
In the equatorial and full Bloch-sphere limits, this yields exact limiting robustness thresholds for noisy incompatible Hamiltonian settings. 
Together, these results establish exact classical limits for average work and identify Hamiltonian incompatibility as the thermodynamic resource behind quantum work advantage. 
Free energy fixes each branch; commutativity fixes the classical task average; Hamiltonian incompatibility violates that average-work law.

\begin{figure*}[t]
    \centering
    \includegraphics[width=\textwidth]{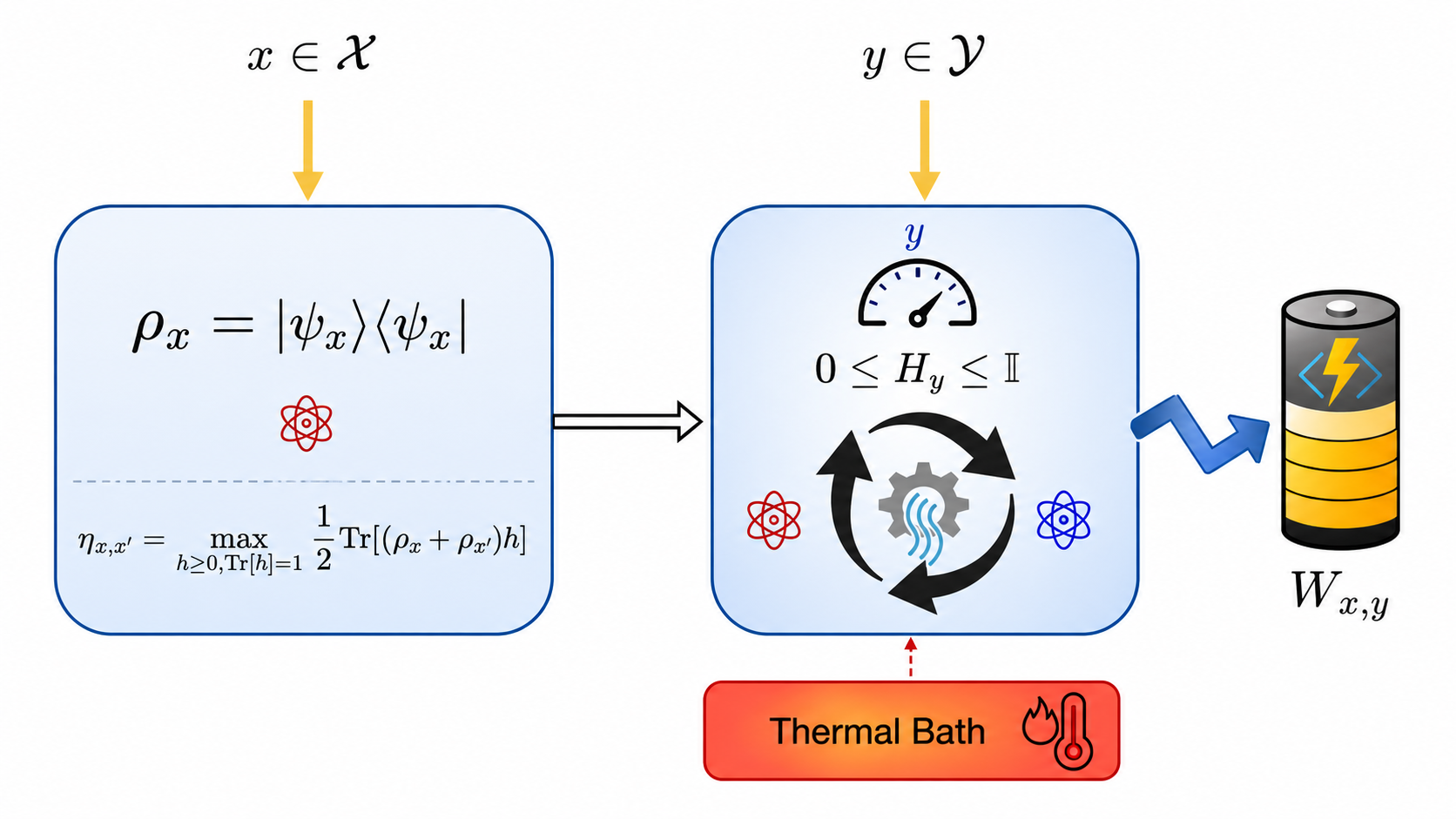}
    \caption{
    \textbf{Work-extraction task.}
    A pure preparation \(\ket{\psi_x}\) is supplied to a single work device, while the source side is constrained only by pairwise maximal-energy constraints under one common normalized Hamiltonian. 
    The device is queried through a bounded Hamiltonian setting \(0\le H_y\le \id\) and may use an arbitrary thermodynamic protocol involving a heat bath, internal control, and work storage. 
    Each run realizes one branch of the task, and repeated runs define the prior-weighted task average. 
    Classicality is imposed only by requiring the Hamiltonian settings of the same device to commute. 
    The paper derives the corresponding classical average-work law and shows that incompatible Hamiltonian settings violate it while every realized branch remains bounded by its own free-energy maximum.
    }
    \label{fig:work-extraction-task}
\end{figure*}

\section{Thermodynamic work extraction and the branch-wise free-energy bound}
\label{sec:absolute-law}

For a fixed thermodynamic process, the maximum extractable work question is already closed by free energy unless additional microscopic structure is introduced. 
Once the input state, Hamiltonian and bath temperature are specified, quantum thermodynamics fixes the maximum average work through the corresponding free-energy difference \cite{Alicki1979,Spohn1978}. 
Any microscopic protocol can at best attain this value reversibly \cite{AlickiHorodecki2004}. 
The role of this section is to recall this single-process law and to isolate why the advantage studied here must concern a task: it considers the average work realized by one device across several supplied preparations and Hamiltonian settings.

The first law relates the change in internal energy \(\Delta U\) of a system to the heat \(\Delta Q\) absorbed from the environment and the work \(\Delta W_{\rm ext}\) extracted through controlled external driving, \(\Delta U=\Delta Q-\Delta W_{\rm ext}\). 
For a quantum system with state \(\rho(t)\in\Dens(\mathbb{C}^d)\) and Hamiltonian \(H(t)\), the internal energy is \(U(t)=\Tr[H(t)\rho(t)]\), and an infinitesimal variation gives \(dU=\Tr(dH\,\rho)+\Tr(H\,d\rho)\). 
With the sign convention that extracted work is positive, \(\delta W_{\rm ext}=-\Tr(dH\,\rho)\) and \(\delta Q=\Tr(H\,d\rho)\). 
For a finite process over \(t\in[0,\tau]\), the extracted work is \(W_{\rm ext}=-\int_0^\tau dt\,\Tr[\dot H(t)\rho(t)]\).

When the system is coupled to a thermal bath at temperature \(T\), the relevant state function is the nonequilibrium free energy \(\mathcal{F}_T(\rho,H):=\Tr(\rho H)-T S(\rho)\) \cite{HorodeckiOppenheim2013,SkrzypczykShortPopescu2014,BrandaoHorodeckiNgOppenheimWehner2015}, where \(S(\rho)=-\Tr(\rho\log\rho)\). 
For a Hamiltonian \(H\), the Gibbs state is \(\tau_H=e^{-\beta H}/Z_H\), with partition function \(Z_H=\Tr(e^{-\beta H})\), \(k_B=1\), and \(\beta=1/T\). 
The Gibbs state minimizes the free energy, so \(\mathcal{F}_T(\rho,H)\ge \mathcal{F}_T(\tau_H,H)=-T\log Z_H\). 
Equivalently, \(\mathcal{F}_T(\rho,H)-\mathcal{F}_T(\tau_H,H)=T\,D(\rho\|\tau_H)\), where \(D(\rho\|\sigma)=\Tr[\rho(\log\rho-\log\sigma)]\) is the quantum relative entropy.

We use the standard upper bound on the average extractable work at fixed bath temperature \cite{HorodeckiOppenheim2013,SkrzypczykShortPopescu2014,Aberg2013,RichensMasanes2016,WatanabeTakagi2026}. 
For a system initially in state \(\rho\), with initial Hamiltonian \(H_i\), bath temperature \(T\), and final equilibrium reference state \(\tau_{H_f}\) associated with a final Hamiltonian \(H_f\), the extracted average work satisfies
\begin{equation}\label{eq:max_work}
W_{\rm ext}\leq \mathcal{F}_T(\rho,H_i)-\mathcal{F}_T(\tau_{H_f},H_f).
\end{equation}
Equivalently, \(W_{\rm ext}=\mathcal{F}_T(\rho,H_i)-\mathcal{F}_T(\tau_{H_f},H_f)-T\Sigma\), where \(\Sigma\ge0\) is the irreversible entropy production of the protocol. 
The bound is saturated in the reversible limit if the final state is \(\tau_{H_f}\). 
For cyclic processes, with \(H_f=H_i=H\), Eq.~\eqref{eq:max_work} gives \(W_{\rm ext}^{\max}=\mathcal{F}_T(\rho,H)-\mathcal{F}_T(\tau_H,H)=T\,D(\rho\|\tau_H)\).

Throughout the task comparison, we use this branch-wise maximum as the thermodynamic envelope. 
The physical device may implement any protocol, reversible or irreversible, and an imperfect protocol can only extract less work on a given branch. 
Thus the classical benchmark derived below is not a bound on a particular implementation strategy. 
It is the most favorable average work compatible with the branch-wise free-energy limit and with mutually commuting Hamiltonian settings. 
Consequently, any observed violation cannot be attributed to protocol inefficiency in the classical device; it certifies that the commuting Hamiltonian description itself is insufficient.

Equation~\eqref{eq:max_work} closes the single-process problem. 
For one fixed state and one fixed Hamiltonian, the reversible work is already determined by free energy difference; one branch cannot witness the separation studied here. 
The work extraction advantage studied in this article begins only when several such branches are assembled into one task using the same work device. 
For a branch specified by a supplied preparation and a Hamiltonian setting, the reversible average work separates into three terms: the internal energy of the preparation under that Hamiltonian, the entropy of the preparation, and the partition-function term of the Hamiltonian. 
The entropy and partition-function terms are branch-local thermodynamic offsets; they do not encode whether different Hamiltonian settings of one device commute. 
The commutativity-sensitive contribution is the collection of internal-energy terms entering the task average. 
The next section assembles these branch quantities into work-extraction tasks. 
Appendices~\ref{app:ergotropy-free-energy} and~\ref{app:free-energy-saturation} record two standard facts used in this separation: ergotropy is bounded by the same free-energy difference, and the reversible free-energy value is approached by an explicit reversible isothermal protocol.

\section{Single-device work-extraction tasks at the branch-wise thermodynamic envelope}
\label{sec:framework}

We now define the task value used in the device comparison. 
For each branch, the actual extracted work of any physical protocol is bounded by the branch-wise free-energy maximum. 
We compare devices at this best-case thermodynamic envelope: branch-local entropy and partition-function terms are fixed by the branch data and do not encode commutativity across Hamiltonian settings. 
The commutativity-sensitive contribution is the internal-energy term \(\Tr(\rho_xH_y)\). 
The task average below is therefore the average of these branch internal-energy contributions, optimized over the admissible Hamiltonian settings of the device.

\begin{definition}[Work-extraction task]
\label{def:work-extraction-task}
A work-extraction task specifies how one device is supplied, queried, and assigned an average work value at the branch-wise thermodynamic envelope. 
In each round, the task samples a branch \((x,y)\) from a prior \(p(x,y)\), supplies the preparation \(\rho_x\), and queries the Hamiltonian setting \(H_y\). 
The branch work value entering the task is the internal-energy contribution
\begin{equation}\label{eq:avg_work_task}
a_{x,y}:=\Tr(\rho_xH_y),
\qquad
\overline W_p:=\sum_{x,y}p(x,y)a_{x,y}.
\end{equation}
Here \(p(x,y)\ge0\) and \(\sum_{x,y}p(x,y)=1\). 
We call \(\overline W_p\) the average work of the task.
\end{definition}

The task specification consists of the supplied preparations, Hamiltonian-setting labels, branch prior, thermodynamic source constraints, and work scale. 
A physical work-extraction protocol realizes one branch in each run and may be imperfect. 
Such imperfections can only reduce the work obtained relative to the branch-wise thermodynamic envelope. 
The benchmark in this paper is deliberately the best-case one: it asks for the largest task average compatible with the source constraints, the bounded Hamiltonian settings, and, for classical devices, mutual commutativity. 
Thus the separation is not a comparison between a good quantum protocol and a bad classical protocol. 
It is a comparison between incompatible Hamiltonian settings and the optimal commuting Hamiltonian description under the same branch-wise thermodynamic limits.

For each \(x\in\cX\), the source supplies a pure preparation \(\rho_x=\ket{\psi_x}\!\bra{\psi_x}\) on an otherwise unrestricted Hilbert space. 
No microscopic source model is assumed. 
The source is specified only through pairwise maximal-energy data. 
For a pair of preparations, one and the same trace-normalized reference Hamiltonian is used for both members of the pair. 
This makes \(\eta_{xx'}\) a thermodynamic relation between two preparations rather than two independent single-branch energy values. 
We define
\begin{equation}\label{eq:pairwise_average_energy}
\eta_{xx'}
:=
\max_{h\ge0,\ \Tr [h]=1}
\frac{1}{2}\Tr\!\left[(\rho_x+\rho_{x'})h\right].
\end{equation}
The factor \(1/2\) makes \(\eta_{xx'}\) the maximal pair-averaged normalized energy attainable by one common Hamiltonian for the two preparations. 
Equivalently, \(\eta_{xx'}=\lambda_{\max}(\rho_x+\rho_{x'})/2\), and for pure preparations \(\eta_{xx'}=(1+|\braket{\psi_x}{\psi_{x'}}|)/2\). 
Thus, for pure preparations, upper and lower constraints on \(\eta_{xx'}\) are exactly upper and lower constraints on the corresponding transition amplitudes. 
The maximization over \(h\) is not an additional physical operation performed in the task; it defines the intrinsic maximal pair energy of the supplied preparations. 
When these pairwise maximal-energy constraints are nontrivial, the source cannot simply encode an independent classical label for every scored branch. 
The source geometry entering the work task is therefore fixed by thermodynamic pair constraints.

For each \(y\in\cY\), the device implements a bounded Hamiltonian setting \(H_y\) satisfying \(0\le H_y\le\id\). 
This boundedness is the work-scale normalization for the device Hamiltonians. 
Upon receiving the preparation \(\rho_x\) and setting label \(y\), the device may implement an arbitrary thermodynamic protocol involving the Hamiltonian setting \(H_y\), a heat bath at temperature \(T\), internal degrees of freedom, and work-storage systems. 
The device may also use an ancilla of arbitrary dimension prepared in a fixed state, say \(\ket{0}\!\bra{0}_A\). 
Since this ancillary state is fixed, we absorb it into the input description and continue to denote by \(\rho_x\) the resulting full input state on the device Hilbert space \(\cH\). 
No bound is imposed on this Hilbert-space dimension, and no microscopic model of the work-extraction mechanism is assumed.

The average work depends only on the supplied preparations and the Hamiltonian settings of the device. 
For a classical work device, the Hamiltonian settings appearing in \(\overline W_p\) must commute with each other. 
The following sections derive the exact upper limits that commutativity imposes on the maximum average work.

\section{Classical work devices and the average-work law}
\label{sec:classical-devices}

We now impose the classical condition on the work device. 
For a work-extraction task, the device is accessed through its Hamiltonian settings. 
The classical condition is that these settings commute with each other. 
This is the only device-side restriction: the Hilbert-space dimension, internal degrees of freedom, ancillas, and microscopic work-extraction protocol remain unrestricted.

\begin{definition}[Classical work device]
\label{def:classical-work-device}
A work device with bounded Hamiltonian settings \(\{H_y\}_{y\in\cY}\), \(0\le H_y\le \id\), is called \emph{classical} for the task if its Hamiltonian settings commute:
\begin{equation}
\label{eq:classical-work-device}
[H_y,H_{y'}]=0
\qquad
\forall\,y,y'\in\cY .
\end{equation}
\end{definition}

Operationally, a classical work device is one whose queried Hamiltonians form a mutually commuting algebra. 
No dimension bound or microscopic engine model is assumed. 
Classicality is imposed only through mutual commutativity of the Hamiltonian settings.

We first consider the minimal nontrivial task, which uses two preparations, denoted \(\ket{\psi_0}\) and \(\ket{\psi_1}\), and two Hamiltonian settings \(H_0,H_1\). 
The source also supplies a reference preparation \(\ket{\psi_\perp}\). 
This reference preparation supplies thermodynamic reference data for the commuting benchmark rather than a scored branch of the task.

The source side of the task is specified by pairwise maximal-energy constraints,
\[
\eta_{0\perp}\le \frac12,
\qquad
\eta_{01}\ge \frac{1+\cos\theta}{2},
\qquad
\eta_{\perp 1}\ge \frac{1+\sin\theta}{2},
\]
with \(0<\theta<\pi/2\). 
These are constraints on intrinsic pair quantities of the supplied preparations, not assumed state-vector equalities. 
Since the preparations are pure and normalized, \(\eta_{0\perp}\ge1/2\) always; hence the first bound enforces orthogonality of \(\ket{\psi_0}\) and \(\ket{\psi_\perp}\). 
The two lower constraints then require transition amplitudes at least \(\cos\theta\) and \(\sin\theta\) with this orthogonal pair. Because \(\cos^2\theta+\sin^2\theta=1\), normalization forces both constraints to be saturated simultaneously. 
The derivation is given in Appendix~\ref{app:main-bound-proof}. 
Thus the two-dimensional source geometry used below is derived from thermodynamic source-side bounds.

For \(i\in\{0,1,\perp\}\) and \(j\in\{0,1\}\), write \(a_{i,j}:=\langle\psi_i|H_j|\psi_i\rangle\). 
The task prior is supported only on the two scored branches \((x,y)=(0,0)\) and \((1,1)\), with probability \(1/2\) assigned to each. 
Hence \(\overline W_p=(a_{0,0}+a_{1,1})/2\). 
The reference preparation is not scored. 
Reference values enter only through the setting averages \(\mu_j:=(a_{0,j}+a_{\perp,j})/2\) and \(R_j:=\sqrt{\mu_j(1-\mu_j)}\), for \(j=0,1\). 
These quantities are retained in the optimization rather than fixed to a balanced work scale.

The following theorem is the classical average-work law for this task: it gives the universal classical average-work restriction imposed by mutual commutativity of the Hamiltonian settings on the scored branch average. 
The law is first stated as a constraint on the work values of any classical implementation. 
The scalar classical benchmark used below is obtained by optimizing this constraint over the admissible work values of classical devices.

\begin{theorem}[Classical average-work law]
\label{thm:commuting-average-bound}
For the three-preparation, two-Hamiltonian task above, every classical work device satisfies the following bound.  For \(0<\mu_0<1\), define \(Z_\theta:=\bigl[\cos(2\theta)(a_{1,0}-\mu_0)-(a_{0,0}-\mu_0)\bigr]/\sin(2\theta)\). Then
\begin{equation}
\label{eq:commuting-average-bound}
\overline W_p
\le
\frac12\left[
a_{0,0}
+
\mu_1
+
R_1\sqrt{1-\frac{Z_\theta^2}{R_0^2}}
\right],
\end{equation}
where \(\mu_j=(a_{0,j}+a_{\perp,j})/2\) and \(R_j=\sqrt{\mu_j(1-\mu_j)}\).  For \(\mu_0\in\{0,1\}\), the endpoint bound is obtained by replacing the square-root factor in Eq.~\eqref{eq:commuting-average-bound} by \(1\).  The bound holds for all bounded commuting Hamiltonian implementations, with no restriction on dimension or microscopic work-extraction mechanism.
\end{theorem}

The proof is given in Appendix~\ref{app:main-bound-proof}. 
This theorem is the classical average-work law for the three-preparation task. It gives a universal constraint on the work values realizable by any single device whose two Hamiltonian settings commute. 
Incompatible Hamiltonian settings can exceed this law while each branch remains governed by its own free-energy maximum.

\section{Optimal average-work advantage from incompatible Hamiltonians}
\label{sec:violation}

A violation of the commuting benchmark demonstrates quantum advantage in the work-extraction task: one device attains a larger average work than any device whose Hamiltonian settings commute. 
We quantify this excess as follows.

\begin{definition}[Average-work advantage]
\label{def:average-work-advantage}
For a fixed work-extraction task, the average-work advantage is \(\Delta_{\rm av}:=\overline W_p^{\,\rm q}-\overline W_p^{\,\rm cl}\), where \(\overline W_p^{\,\rm cl}\) is the optimum over all bounded commuting Hamiltonian implementations, and \(\overline W_p^{\,\rm q}\) is the value attainable by bounded Hamiltonian settings that need not commute.
\end{definition}

For the task of Sec.~\ref{sec:classical-devices}, Theorem~\ref{thm:commuting-average-bound} gives the classical law obeyed by every classical implementation. 
The reference averages \(\mu_0,\mu_1\) are retained as variational parameters rather than fixed to particular values. 
The classical benchmark is obtained by optimizing this law over the admissible work values of classical devices. 
To compare it with unrestricted incompatible Hamiltonians, we use the following one-setting envelope.

\begin{lemma}[Unrestricted one-setting envelope]
\label{lem:one-setting-envelope}
Let \(\ket{\psi_0},\ket{\psi_\perp}\) be an orthogonal reference pair, and let \(\ket{\phi}\) be a pure state in their span. 
Among all bounded Hamiltonians \(0\le H\le\id\) with reference average \(\mu=(\langle\psi_0|H|\psi_0\rangle+\langle\psi_\perp|H|\psi_\perp\rangle)/2\), the largest attainable expectation value on \(\ket{\phi}\) is \(\mu+\min\{\mu,1-\mu\}\). 
The optimization is over arbitrary-dimensional Hamiltonian implementations.
\end{lemma}

The proof is a compression argument and is given in Appendix~\ref{app:one-setting-envelope}. 
Applying Lemma~\ref{lem:one-setting-envelope} independently to the two Hamiltonian settings gives the optimal incompatible value at reference averages \(\mu_0,\mu_1\). 
Writing \(\ell_j:=\min\{\mu_j,1-\mu_j\}\), this gives \(\overline W_p^{\,\rm q}=\frac12(\mu_0+\mu_1+\ell_0+\ell_1)\). 
Appendix~\ref{app:optimal-minimal-advantage} proves the optimal minimal-task advantage and shows that the optimum is attained at the balanced source point. 
The next result optimizes over the full minimal-task family. 
The symmetric source point and balanced reference averages are not assumptions; they are selected by the average-work advantage optimization.

\begin{theorem}[Optimal average-work advantage in the minimal task]
\label{thm:optimal-minimal-average-advantage}
For the minimal three-preparation, two-Hamiltonian task, the maximum average-work advantage is
\begin{equation}
\label{eq:max-average-advantage}
\Delta_{\rm av}^{\max}
=
\frac12\left(1-\frac1{\sqrt2}\right).
\end{equation}
The optimum is attained when the source-side constraints are saturated at the symmetric point: \(\eta_{0\perp}=1/2\) and \(\eta_{01}=\eta_{\perp 1}=\frac12(1+1/\sqrt2)\), equivalently \(\theta=\pi/4\), together with balanced reference averages \(\mu_0=\mu_1=1/2\). 
At this optimum, \(\overline W_p^{\,\rm q}=1\) and \(\overline W_p^{\,\rm cl}=\frac12(1+1/\sqrt2)\). 
The optimal incompatible Hamiltonian settings are the rank-one Hamiltonians \(H_0=\ket{\psi_0}\!\bra{\psi_0}\) and \(H_1=\ket{\psi_1}\!\bra{\psi_1}\).
\end{theorem}

The proof is given in Appendix~\ref{app:optimal-minimal-advantage}. 
The theorem shows that the symmetric saturation of the source-side constraints and the balanced reference averages are selected by the average-work optimization itself. 
They are outputs of the optimization rather than inputs to the classical benchmark. 
The next question is robustness: how much noise can the Hamiltonian settings tolerate before the advantage disappears? 
We first analyze the robustness of the minimal setup, then develop a hierarchy of work-extraction tasks in which larger source families and more Hamiltonian settings make the incompatibility-based average-work violation increasingly robust.

\section{Robustness and hierarchy of average-work tasks}
\label{sec:robustness-hierarchy}

The optimized minimal task gives an exact separation for ideal Hamiltonian settings. 
We now ask whether this separation survives imperfect settings. 
The work-extraction task itself is kept fixed: the supplied preparations remain pure, obey the same pairwise maximal-energy source constraints, and are sampled with the same task prior. 
Only the Hamiltonian settings used by the incompatible implementation are degraded.

We analyze the explicit qubit implementation that attains the advantage. 
For this implementation, each target rank-one setting is replaced by the depolarized setting \(H_x^{(v)}=v\ket{\psi_x}\!\bra{\psi_x}+(1-v)\id_2/2\), with \(0\le v\le1\). 
Here \(\id_2\) is the identity on the two-dimensional source span. 
This depolarized Hamiltonian setting interpolates between the target rank-one setting and the state-independent setting. 
For a scored preparation paired with its corresponding Hamiltonian setting, the average work decreases linearly with the visibility.

For the optimized minimal task of Theorem~\ref{thm:optimal-minimal-average-advantage}, the noisy qubit implementation reaches the commuting benchmark at \(v_c=1/\sqrt2\). 
Thus the minimal average-work advantage survives precisely above this visibility.

The minimal task is the first member of a qubit hierarchy. 
The hierarchy is built from pure qubit preparations associated with a finite family of Bloch-sphere directions. 
The task prior is uniform over the matched scored pairs \((x,x)\). 
The source is specified by the pairwise maximal-energy data of this pure preparation family. 
The reference data are chosen with the same average work scale for every setting, so that the hierarchy compares one fixed source family against one device.

\begin{theorem}[Hierarchy threshold from classical alignment]
\label{thm:hierarchy-threshold}
Let \(M_N=\{m_x\}_{x=1}^N\) be the Bloch-sphere directions of the source family. 
For the balanced qubit work-extraction task with uniform prior over the matched scored pairs \((x,x)\), define \(R(M_N):=\max_{\|u\|=1}N^{-1}\sum_{x=1}^N |u\cdot m_x|\). 
The exact classical benchmark is \(\overline W_p^{\rm cl}=\frac12(1+R(M_N))\). 
The noisy incompatible qubit implementation has \(\overline W_p^{\rm q}(v)=\frac12(1+v)\), and is advantageous exactly when \(v>R(M_N)\). 
Equivalently, the visibility threshold is \(v_c(M_N)=R(M_N)\).
\end{theorem}

The proof is given in Appendix~\ref{app:hierarchy-threshold}. 
The quantity \(R(M_N)\) is the best average alignment available to any classical work device for the task. 
Robustness is therefore controlled by how poorly any mutually commuting Hamiltonian family can align with the whole source family.

The minimal task gives the first separation, with visibility threshold \(v_c=1/\sqrt2\). 
The hierarchy improves this threshold by enlarging the thermodynamic source family. 
For evenly spaced equatorial source families, the limiting threshold is \(v_c=2/\pi\). 
For source families approaching the full Bloch-sphere limit, the threshold reaches \(v_c=1/2\), so the ideal full Bloch-sphere limiting average-work gap is \(1/4\). 
Thus robustness is controlled by a geometric obstruction: a classical work device must align one commuting Hamiltonian family with the whole source family, while incompatible Hamiltonians can align setting by setting.

\section{Discussion}
\label{sec:discussion}

We have derived a thermodynamic law for average work at the level of a task. 
For a single branch, there is no quantum-over-classical separation of the kind studied here: once the state, Hamiltonian, and bath are fixed, free energy fixes the reversible work. 
The new question begins only when one device is used across several supplied preparations and Hamiltonian settings. 
We compare devices at the branch-wise thermodynamic envelope, so every branch is assigned its best free-energy-limited value. 
The resulting separation is therefore not a comparison between an efficient quantum protocol and an inefficient classical one. 
It is a limitation imposed by mutual commutativity of the Hamiltonian settings themselves.

The form of the law is forced by thermodynamics. 
A meaningful average-work bound cannot be stated only in terms of the scored branches. 
A larger value could otherwise come from a larger source energy scale, from branch information encoded in the preparations, or from arbitrary shifts of the Hamiltonian settings. 
The reference quantities appearing in the classical law are therefore not auxiliary parameters. 
They are the thermodynamic benchmark data that make the comparison meaningful: the source-side energy constraints determine the allowed source geometry, while the setting reference values fix the work scale. 
Once these data are fixed, commutativity imposes the classical average-work bound.

The source condition used here is purely thermodynamic. 
The source is specified through pairwise maximal-energy constraints under one common normalized Hamiltonian for each pair of preparations. 
For pure preparations, these constraints determine the transition amplitudes and hence the geometry entering the task. 
Thus both sides of the comparison are expressed in the same physical currency: normalized Hamiltonian-energy expectations. 
This is not a dimension assumption, not a distinguishability assumption, and not an operational-equivalence assumption written in different notation. 
Its role is structurally analogous to source constraints in prepare-and-measure theory, where nontrivial black-box classical limits require restrictions on what the source can carry, such as indistinguishability, parity-obliviousness, operational classicality, or semidevice-independent source/channel assumptions \cite{ChaturvediSaha2020QuantumPrescriptions,SahaChaturvedi2019PreparationContextuality,Sarkar2023OperationalClassicality,SarkarDatta2025SDIResources,RajPrasadChaturvedi2026WPD,ChaturvediPawlowskiSaha2026Epistemic}. 
Here the restriction is energetic from the start: it is a trusted thermodynamic constraint on the supplied preparations, while the work device itself remains unrestricted except for the classical compatibility condition being tested.

Mathematically, before the pure-state reduction is used, the source-device problem belongs to the landscape of tracial noncommuting polynomial optimization: one optimizes trace polynomials in noncommuting positive operators subject to normalization, boundedness, thermodynamic source constraints, and, for the classical benchmark, commutativity of the Hamiltonian settings. 
Such noncommuting-operator optimization is the standard backbone behind semidefinite approaches to quantum contextuality and prepare-and-measure correlations \cite{ChaturvediFarkasWright2021ContextualBehaviours,TavakoliCruzeiroUolaAbbott2021ContextualCorrelations,HazraSahaChaturvediBeraMajumdar2026NoncontextualPolytopes}. 
The present pure-state thermodynamic instance is special because it does not remain a numerical hierarchy. 
The pairwise maximal-energy constraints collapse the source geometry exactly, and the commuting average-work law can be solved in closed form.

Purity is the remaining source-side boundary of the present law. 
It is not a restriction on the work device, which remains arbitrary-dimensional and microscopically unrestricted. 
Its role is to make the pairwise thermodynamic data complete: for pure preparations, maximal pairwise energy fixes the transition amplitudes. 
For mixed preparations, the same data no longer determine the source geometry, so a stronger thermodynamic source primitive is needed. 
Removing purity while keeping the source assumptions energetic, rather than dimensional or informational, is the natural next step.

The work device is treated at the same primitive level. 
Each query supplies one bounded Hamiltonian setting \(0\le H_y\le\id\), and the task uses the corresponding energy expectation \(\Tr(\rho_xH_y)\). 
No engine model, dimension bound, auxiliary outcome structure, or microscopic protocol is imposed. 
The benchmark is the optimum over all such devices satisfying the same thermodynamic reference data and the classical compatibility condition. 
This is a trusted-source, work-side black-box comparison: the preparation side is constrained by thermodynamic energy data, while the Hamiltonian side is optimized over arbitrary-dimensional implementations, in the same methodological spirit in which measurement-device-independent protocols keep the measurement side uncharacterized while retaining a trusted preparation layer \cite{ChaturvediBanik2015MDIRandomness,ArgillanderSpegelLexne2025MDIQRNG}. 
This is why the result is a single-device thermodynamic law rather than a model-dependent engine inequality.

The bounded-Hamiltonian assumption is a normalization of the work scale, not a microscopic spectral restriction. 
For any lower-bounded Hamiltonian with a finite task-relevant energy scale, one may subtract the irrelevant energy origin and divide the work, heat, and internal-energy contribution of the branch by that reference scale. 
This maps the comparison to the normalized form \(0\le H_y\le\id\) on the energy window used by the task, while preserving the question of converting supplied energy into extractable work. 
If no finite reference scale is fixed, no finite dimensionless average-work law can compare two devices in the first place; the normalization is therefore part of stating a thermodynamic limit, not an extra model assumption.

The connection to communication and prepare-and-measure tasks is precise, but it is one-way unless extra operational structure is added. 
Our thermodynamic primitive is smaller than a full communication experiment: for each setting \(y\), the work device supplies one bounded Hamiltonian \(0\le H_y\le\id\), and the task uses only the energy expectation \(\Tr(\rho_xH_y)\). 
Thus, in a black-box communication reading, each thermodynamic setting corresponds only to one scored binary receiver test on Bob's side. 
The complementary operator \(\id-H_y\) may be introduced if one wants to embed the thermodynamic number into a normalized binary measurement, but this complement is not used in the work law and carries no independent thermodynamic role in the theorem. 
From the thermodynamic task one obtains a single scored receiver event per setting, not an arbitrary multi-outcome measurement, not an instrument, and not a full prepare-and-measure behavior. 
Conversely, mapping a general communication or prepare-and-measure task into a thermodynamic work task requires additional assumptions that are not thermodynamic primitives. 
In contextual prepare-and-measure formulations, these assumptions are supplied by operational equivalences among preparations and measurement effects \cite{SchmidSpekkensWolfe2018,ChaturvediFarkasWright2021ContextualBehaviours,SahaChaturvedi2019PreparationContextuality}. 
In semidevice-independent resource-detection formulations, one instead fixes a prepare-and-measure architecture, such as a channel, dimension, or instrument structure connecting preparation and measurement devices \cite{SarkarDatta2025SDIResources}. 
Those frameworks are useful, but they are stronger than the primitive work setting considered here. 
The average-work law is derived before such completions are imposed: it uses source energy constraints, bounded Hamiltonian settings, and commutativity alone.

The hierarchy shows that the advantage is robust under imperfect Hamiltonian settings. 
The minimal task gives the first separation, with visibility threshold \(v_c=1/\sqrt2\). 
For evenly distributed equatorial source families, the limiting threshold becomes \(2/\pi\). 
In the full Bloch-sphere limit, it reaches \(1/2\), equivalently an ideal limiting average-work gap of \(1/4\). 
The mechanism is simultaneous alignment: a commuting device must align one compatible family of Hamiltonian settings with the whole source family, while incompatible Hamiltonians can align setting by setting. 
Increasing the source family therefore strengthens the witness of Hamiltonian incompatibility.

This resource is distinct from coherence relative to a fixed Hamiltonian. 
Coherence in one Hamiltonian eigenbasis is already known to contribute to nonequilibrium thermodynamics and entropy production \cite{LostaglioJenningsRudolph2015,LostaglioKorzekwaJenningsRudolph2015,SantosCeleriLandiPaternostro2019}. 
Here the resource is incompatibility among the Hamiltonian settings accessible to one work device. 
It is witnessed against the full class of arbitrary-dimensional commuting implementations, while every branch remains bounded by its own free-energy law. 
The endpoint is a device-independent thermodynamic statement: a work advantage certified from correlations alone.\footnote{
This Bell-certified formulation will appear on arXiv in a separate manuscript. 
It uses the thermodynamic law developed here as the primitive work benchmark and upgrades the certification layer from source-trusted thermodynamic constraints to observed correlations, yielding a theory-independent certification of thermodynamic work advantage.
}

\section*{Acknowledgements}
We thank Borhan Ahmadi, Tanmoy Biswas, Shubhayan Sarkar, Micha\l{} Horodecki, Manik Banik, Pawe\l{} Mazurek, and Armin Tavakoli for insightful discussions. 
S.R., A.B.R., and P.H. acknowledge support from the IRA Programme, project no. FENG.02.01-IP.05-0006/23, financed by the FENG programme 2021-2027, Priority FENG.02, Measure FENG.02.01, with the support of the Foundation for Polish Science. 
A.C. acknowledges support from the KLAR Grant No. BNI/PST/2023/1/00013/U/00001, funded by NAWA.

\bibliographystyle{apsrev4-2}
\bibliography{main}

@article{Alicki1979,
  author  = {Alicki, Robert},
  title   = {The quantum open system as a model of the heat engine},
  journal = {Journal of Physics A: Mathematical and General},
  volume  = {12},
  number  = {5},
  pages   = {L103--L107},
  year    = {1979},
  doi     = {10.1088/0305-4470/12/5/007},
  url     = {https://doi.org/10.1088/0305-4470/12/5/007}
}

@article{Spohn1978,
  author  = {Spohn, Herbert},
  title   = {Entropy production for quantum dynamical semigroups},
  journal = {Journal of Mathematical Physics},
  volume  = {19},
  number  = {5},
  pages   = {1227--1230},
  year    = {1978},
  doi     = {10.1063/1.523789},
  url     = {https://doi.org/10.1063/1.523789}
}

@misc{AlickiHorodecki2004,
  author        = {Alicki, Robert and Horodecki, Micha{\l} and Horodecki, Pawe{\l} and Horodecki, Ryszard},
  title         = {Thermodynamics of quantum informational systems: Hamiltonian description},
  year          = {2004},
  eprint        = {quant-ph/0402012},
  archivePrefix = {arXiv},
  primaryClass  = {quant-ph},
  doi           = {10.48550/arXiv.quant-ph/0402012},
  url           = {https://arxiv.org/abs/quant-ph/0402012}
}

@article{HorodeckiOppenheim2013,
  author  = {Horodecki, Micha{\l} and Oppenheim, Jonathan},
  title   = {Fundamental limitations for quantum and nanoscale thermodynamics},
  journal = {Nature Communications},
  volume  = {4},
  pages   = {2059},
  year    = {2013},
  doi     = {10.1038/ncomms3059},
  url     = {https://doi.org/10.1038/ncomms3059}
}

@article{SkrzypczykShortPopescu2014,
  author  = {Skrzypczyk, Paul and Short, Anthony J. and Popescu, Sandu},
  title   = {Work extraction and thermodynamics for individual quantum systems},
  journal = {Nature Communications},
  volume  = {5},
  pages   = {4185},
  year    = {2014},
  doi     = {10.1038/ncomms5185},
  url     = {https://doi.org/10.1038/ncomms5185}
}

@article{Aberg2013,
  author  = {{\AA}berg, Johan},
  title   = {Truly work-like work extraction via a single-shot analysis},
  journal = {Nature Communications},
  volume  = {4},
  pages   = {1925},
  year    = {2013},
  doi     = {10.1038/ncomms2712},
  url     = {https://doi.org/10.1038/ncomms2712}
}

@article{BrandaoHorodeckiNgOppenheimWehner2015,
  author  = {Brand{\~a}o, Fernando G. S. L. and Horodecki, Micha{\l} and Ng, Nelly H. Y. and Oppenheim, Jonathan and Wehner, Stephanie},
  title   = {The second laws of quantum thermodynamics},
  journal = {Proceedings of the National Academy of Sciences},
  volume  = {112},
  number  = {11},
  pages   = {3275--3279},
  year    = {2015},
  doi     = {10.1073/pnas.1411728112},
  url     = {https://doi.org/10.1073/pnas.1411728112}
}

@article{RichensMasanes2016,
  author  = {Richens, Jonathan G. and Masanes, Llu{\'i}s},
  title   = {Work extraction from quantum systems with bounded fluctuations in work},
  journal = {Nature Communications},
  volume  = {7},
  pages   = {13511},
  year    = {2016},
  doi     = {10.1038/ncomms13511},
  url     = {https://doi.org/10.1038/ncomms13511}
}

@article{SafranekRosaBinder2023,
  author  = {{\v{S}}afr{\'a}nek, Dominik and Rosa, Dario and Binder, Felix C.},
  title   = {Work Extraction from Unknown Quantum Sources},
  journal = {Physical Review Letters},
  volume  = {130},
  pages   = {210401},
  year    = {2023},
  doi     = {10.1103/PhysRevLett.130.210401},
  url     = {https://doi.org/10.1103/PhysRevLett.130.210401}
}

@article{WatanabeTakagi2026,
  author  = {Watanabe, Kaito and Takagi, Ryuji},
  title   = {Universal work extraction in quantum thermodynamics},
  journal = {Nature Communications},
  volume  = {17},
  pages   = {1857},
  year    = {2026},
  doi     = {10.1038/s41467-026-69143-3},
  url     = {https://doi.org/10.1038/s41467-026-69143-3}
}

@article{PuszWoronowicz1978,
  author  = {Pusz, Wies{\l}aw and Woronowicz, Stanis{\l}aw Lech},
  title   = {Passive states and {KMS} states for general quantum systems},
  journal = {Communications in Mathematical Physics},
  volume  = {58},
  number  = {3},
  pages   = {273--290},
  year    = {1978},
  doi     = {10.1007/BF01614224},
  url     = {https://doi.org/10.1007/BF01614224}
}

@book{BreuerPetruccione2002,
  author    = {Breuer, Heinz-Peter and Petruccione, Francesco},
  title     = {The Theory of Open Quantum Systems},
  publisher = {Oxford University Press},
  address   = {Oxford},
  year      = {2002},
}

@article{Lenard1978,
  author  = {Lenard, Andrew},
  title   = {Thermodynamical proof of the {Gibbs} formula for elementary quantum systems},
  journal = {Journal of Statistical Physics},
  volume  = {19},
  number  = {6},
  pages   = {575--586},
  year    = {1978},
  doi     = {10.1007/BF01011769},
  url     = {https://doi.org/10.1007/BF01011769}
}

@article{AllahverdyanBalianNieuwenhuizen2004,
  author  = {Allahverdyan, A. E. and Balian, R. and Nieuwenhuizen, Th. M.},
  title   = {Maximal work extraction from finite quantum systems},
  journal = {Europhysics Letters},
  volume  = {67},
  number  = {4},
  pages   = {565--571},
  year    = {2004},
  doi     = {10.1209/epl/i2004-10101-2},
  url     = {https://doi.org/10.1209/epl/i2004-10101-2}
}

@article{FrancicaBinderGuarnieriMitchisonGooldPlastina2020,
  author  = {Francica, Gianluca and Binder, Felix C. and Guarnieri, Giacomo and Mitchison, Mark T. and Goold, John and Plastina, Francesco},
  title   = {Quantum Coherence and Ergotropy},
  journal = {Physical Review Letters},
  volume  = {125},
  pages   = {180603},
  year    = {2020},
  doi     = {10.1103/PhysRevLett.125.180603},
  url     = {https://doi.org/10.1103/PhysRevLett.125.180603}
}

@article{BiswasLobejkoMazurekJalowieckiHorodecki2022,
  author  = {Biswas, Tanmoy and {\L}obejko, Marcin and Mazurek, Pawe{\l} and Ja{\l}owiecki, Konrad and Horodecki, Micha{\l}},
  title   = {Extraction of ergotropy: free energy bound and application to open cycle engines},
  journal = {Quantum},
  volume  = {6},
  pages   = {841},
  year    = {2022},
  doi     = {10.22331/q-2022-10-17-841},
  url     = {https://doi.org/10.22331/q-2022-10-17-841}
}

@article{BiswasDattaGarciaPintos2025,
  author  = {Biswas, Tanmoy and Datta, Chandan and Garc{\'i}a-Pintos, Luis Pedro},
  title   = {Quantum Thermodynamic Advantage in Work Extraction from Steerable Quantum Correlations},
  journal = {Physical Review Letters},
  volume  = {135},
  pages   = {110402},
  year    = {2025},
  doi     = {10.1103/9qcc-7lq5},
  url     = {https://doi.org/10.1103/9qcc-7lq5}
}

@misc{HsiehGessner2024General,
  author        = {Hsieh, Chung-Yun and Gessner, Manuel},
  title         = {General quantum resources provide advantages in work extraction tasks},
  year          = {2024},
  eprint        = {2403.18753},
  archivePrefix = {arXiv},
  primaryClass  = {quant-ph},
  doi           = {10.48550/arXiv.2403.18753},
  url           = {https://arxiv.org/abs/2403.18753}
}

@misc{GarciaPintosBiswasDatta2026Robustness,
  author        = {Garc{\'i}a-Pintos, Luis Pedro and Biswas, Tanmoy and Datta, Chandan},
  title         = {Robustness as a thermodynamic currency: work advantages and preparation costs of nonclassical states},
  year          = {2026},
  eprint        = {2603.04618},
  archivePrefix = {arXiv},
  primaryClass  = {quant-ph},
  doi           = {10.48550/arXiv.2603.04618},
  url           = {https://arxiv.org/abs/2603.04618}
}

@misc{BiswasDattaGarciaPintosCooling2025,
  author        = {Biswas, Tanmoy and Datta, Chandan and Garc{\'i}a-Pintos, Luis Pedro},
  title         = {All steerable quantum correlations can provide thermodynamic advantages in cooling},
  year          = {2025},
  eprint        = {2511.16999},
  archivePrefix = {arXiv},
  primaryClass  = {quant-ph},
  doi           = {10.48550/arXiv.2511.16999},
  url           = {https://arxiv.org/abs/2511.16999}
}

@article{Sarkar2023OperationalClassicality,
  author  = {Sarkar, Shubhayan},
  title   = {An Operational Notion of Classicality Based on Physical Principles},
  journal = {Foundations of Physics},
  volume  = {53},
  pages   = {47},
  year    = {2023},
  doi     = {10.1007/s10701-023-00687-w},
  url     = {https://doi.org/10.1007/s10701-023-00687-w}
}

@article{SarkarDatta2025SDIResources,
  author  = {Sarkar, Shubhayan and Datta, Chandan},
  title   = {Detecting quantum resources in a semi-device-independent framework},
  journal = {Physical Review A},
  volume  = {111},
  pages   = {L040402},
  year    = {2025},
  doi     = {10.1103/PhysRevA.111.L040402},
  url     = {https://doi.org/10.1103/PhysRevA.111.L040402}
}

@article{ChaturvediSaha2020QuantumPrescriptions,
  author    = {Chaturvedi, Anubhav and Saha, Debashis},
  title     = {Quantum prescriptions are ontologically more distinct than they are operationally distinguishable},
  journal   = {Quantum},
  volume    = {4},
  pages     = {345},
  year      = {2020},
  doi       = {10.22331/q-2020-10-21-345},
  url       = {https://doi.org/10.22331/q-2020-10-21-345},
  publisher = {Verein zur F{\"o}rderung des Open Access Publizierens in den Quantenwissenschaften}
}

@article{ChaturvediPawlowskiSaha2026Epistemic,
  author  = {Chaturvedi, Anubhav and Paw{\l}owski, Marcin and Saha, Debashis},
  title   = {Epistemic incompleteness of quantum theory},
  journal = {Physical Review A},
  volume  = {113},
  pages   = {042445},
  year    = {2026},
  doi     = {10.1103/wqhx-qfv1},
  url     = {https://doi.org/10.1103/wqhx-qfv1}
}

@article{LostaglioJenningsRudolph2015,
  author  = {Lostaglio, Matteo and Jennings, David and Rudolph, Terry},
  title   = {Description of quantum coherence in thermodynamic processes requires constraints beyond free energy},
  journal = {Nature Communications},
  volume  = {6},
  pages   = {6383},
  year    = {2015},
  doi     = {10.1038/ncomms7383},
  url     = {https://doi.org/10.1038/ncomms7383}
}

@book{KurizkiKofman2021,
  author    = {Kurizki, Gershon and Kofman, Abraham G.},
  title     = {Thermodynamics and Control of Open Quantum Systems},
  publisher = {Cambridge University Press},
  address   = {Cambridge},
  year      = {2021},
  doi       = {10.1017/9781316798454},
  isbn      = {9781107175419},
}

@article{LostaglioKorzekwaJenningsRudolph2015,
  author  = {Lostaglio, Matteo and Korzekwa, Kamil and Jennings, David and Rudolph, Terry},
  title   = {Quantum Coherence, Time-Translation Symmetry, and Thermodynamics},
  journal = {Physical Review X},
  volume  = {5},
  pages   = {021001},
  year    = {2015},
  doi     = {10.1103/PhysRevX.5.021001},
  url     = {https://doi.org/10.1103/PhysRevX.5.021001}
}

@article{SantosCeleriLandiPaternostro2019,
  author  = {Santos, Jader P. and C{\'e}leri, Lucas C. and Landi, Gabriel T. and Paternostro, Mauro},
  title   = {The role of quantum coherence in non-equilibrium entropy production},
  journal = {npj Quantum Information},
  volume  = {5},
  pages   = {23},
  year    = {2019},
  doi     = {10.1038/s41534-019-0138-y},
  url     = {https://doi.org/10.1038/s41534-019-0138-y}
}

@article{SchmidSpekkensWolfe2018,
  author  = {Schmid, David and Spekkens, Robert W. and Wolfe, Elie},
  title   = {All the noncontextuality inequalities for arbitrary prepare-and-measure experiments with respect to any fixed set of operational equivalences},
  journal = {Physical Review A},
  volume  = {97},
  pages   = {062103},
  year    = {2018},
  doi     = {10.1103/PhysRevA.97.062103},
  url     = {https://doi.org/10.1103/PhysRevA.97.062103}
}

@article{ChaturvediFarkasWright2021ContextualBehaviours,
  author  = {Chaturvedi, Anubhav and Farkas, M{\'a}t{\'e} and Wright, Victoria J.},
  title   = {Characterising and bounding the set of quantum behaviours in contextuality scenarios},
  journal = {Quantum},
  volume  = {5},
  pages   = {484},
  year    = {2021},
  doi     = {10.22331/q-2021-06-29-484},
  url     = {https://doi.org/10.22331/q-2021-06-29-484}
}

@article{TavakoliCruzeiroUolaAbbott2021ContextualCorrelations,
  author  = {Tavakoli, Armin and de Gois, Caique J. and Uola, Roope and Abbott, Alastair A.},
  title   = {Bounding and simulating contextual correlations in quantum theory},
  journal = {PRX Quantum},
  volume  = {2},
  pages   = {020334},
  year    = {2021},
  doi     = {10.1103/PRXQuantum.2.020334},
  url     = {https://doi.org/10.1103/PRXQuantum.2.020334}
}

@article{SahaChaturvedi2019PreparationContextuality,
  author  = {Saha, Debashis and Chaturvedi, Anubhav},
  title   = {Preparation contextuality as an essential feature underlying quantum communication advantage},
  journal = {Physical Review A},
  volume  = {100},
  pages   = {022108},
  year    = {2019},
  doi     = {10.1103/PhysRevA.100.022108},
  url     = {https://doi.org/10.1103/PhysRevA.100.022108}
}

@article{RajPrasadChaturvedi2026WPD,
  author  = {Raj, Chithra and Prasad, Tushita and Chaturvedi, Anubhav and Pollyceno, Lucas and Spegel-Lexne, Daniel and G{\'o}mez, Santiago and Argillander, Joakim and Alarc{\'o}n, Alvaro and Xavier, Guilherme B. and Paw{\l}owski, Marcin and Dieguez, Pedro R.},
  title   = {Certifying semi-device-independent security via wave-particle duality experiments},
  journal = {npj Quantum Information},
  volume  = {12},
  pages   = {7},
  year    = {2026},
  doi     = {10.1038/s41534-025-01160-4},
  url     = {https://doi.org/10.1038/s41534-025-01160-4}
}

@article{HazraSahaChaturvediBeraMajumdar2026NoncontextualPolytopes,
  author  = {Hazra, Soumyabrata and Saha, Debashis and Chaturvedi, Anubhav and Bera, Subhankar and Majumdar, A. S.},
  title   = {Efficient Computation of Generalized Noncontextual Polytopes and Quantum Violation of Their Facet Inequalities},
  journal = {Quantum},
  volume  = {10},
  pages   = {2015},
  year    = {2026},
  doi     = {10.22331/q-2026-03-09-2015},
  url     = {https://doi.org/10.22331/q-2026-03-09-2015}
}

@article{ChaturvediBanik2015MDIRandomness,
  author  = {Chaturvedi, Anubhav and Banik, Manik},
  title   = {Measurement-device-independent randomness from local entangled states},
  journal = {EPL (Europhysics Letters)},
  volume  = {112},
  pages   = {30003},
  year    = {2015},
  doi     = {10.1209/0295-5075/112/30003},
  url     = {https://doi.org/10.1209/0295-5075/112/30003}
}

@misc{ArgillanderSpegelLexne2025MDIQRNG,
  author        = {Argillander, Joakim and Spegel-Lexne, Daniel and Clason, Martin and Dieguez, Pedro R. and Paw{\l}owski, Marcin and Chaturvedi, Anubhav and Xavier, Guilherme B.},
  title         = {High-dimensional detection-loophole-free measurement-device-independent quantum random number generator},
  year          = {2025},
  eprint        = {2510.06317},
  archivePrefix = {arXiv},
  primaryClass  = {quant-ph},
  doi           = {10.48550/arXiv.2510.06317},
  url           = {https://arxiv.org/abs/2510.06317}
}

\clearpage
\appendix

\section{Free-energy bound on ergotropy}
\label{app:ergotropy-free-energy}

This appendix records why unitary work extraction cannot exceed the nonequilibrium free-energy difference \cite{PuszWoronowicz1978,Lenard1978,AllahverdyanBalianNieuwenhuizen2004,BiswasLobejkoMazurekJalowieckiHorodecki2022}. 
The point is that ergotropy compares \(\rho\) with its passive rearrangement \cite{AllahverdyanBalianNieuwenhuizen2004,FrancicaBinderGuarnieriMitchisonGooldPlastina2020}, while the free-energy benchmark compares \(\rho\) with the Gibbs state; Gibbs minimization makes the latter bound larger. 
This relation separates the single-branch free-energy bound from the task-level incompatibility effect studied in the main text.

\begin{lemma}[Free-energy bound on ergotropy]
\label{lem:ergotropy-free-energy}
Let \(\rho\) be a quantum state, let \(H\) be a Hamiltonian, and let
\[
\tau=\frac{e^{-H/T}}{\Tr(e^{-H/T})}
\]
be the Gibbs state at temperature \(T>0\). Define
\[
\mathcal F_T(\sigma,H):=\Tr(\sigma H)-T S(\sigma),
\qquad
S(\sigma):=-\Tr(\sigma\log\sigma).
\]
The ergotropy
\[
\mathcal E(\rho,H):=\Tr(\rho H)-\min_U\Tr(U\rho U^\dagger H)
\]
satisfies
\[
\mathcal E(\rho,H)\le \mathcal F_T(\rho,H)-\mathcal F_T(\tau,H).
\]
\end{lemma}

\begin{proof}
Let \(\rho_p=U_*\rho U_*^\dagger\) be a passive state of \(\rho\) with respect to \(H\), where \(U_*\) attains
\[
\Tr(\rho_pH)=\min_U\Tr(U\rho U^\dagger H).
\]
Then
\[
\mathcal E(\rho,H)=\Tr(\rho H)-\Tr(\rho_pH).
\]
Since \(\rho_p\) is unitarily related to \(\rho\), one has \(S(\rho_p)=S(\rho)\). 
The entropy terms therefore cancel when comparing \(\rho\) with its passive rearrangement, giving
\[
\mathcal F_T(\rho,H)-\mathcal F_T(\rho_p,H)
=
\Tr(\rho H)-\Tr(\rho_pH)
=
\mathcal E(\rho,H).
\]
The only thermodynamic input is that the Gibbs state minimizes \(\mathcal F_T(\cdot,H)\) at fixed \(H\) and \(T\), so \(\mathcal F_T(\tau,H)\le \mathcal F_T(\rho_p,H)\). Therefore
\[
\mathcal E(\rho,H)
=
\mathcal F_T(\rho,H)-\mathcal F_T(\rho_p,H)
\le
\mathcal F_T(\rho,H)-\mathcal F_T(\tau,H),
\]
as claimed.
\end{proof}

\section{Protocol approaching the branch-wise free-energy bound}
\label{app:free-energy-saturation}

We describe an explicit reversible-limit protocol that approaches the free-energy benchmark for a single branch. 
This supports the statement that the branch-wise problem is closed by free energy; the advantage in the main text is therefore necessarily task-level. 
The proof has four steps: construct a full-rank Gibbs approximation to the input state, quench to its Hamiltonian, thermalize, and reversibly return to the branch Hamiltonian.

Fix one branch and suppress the labels \((x,y)\). The input state is \(\rho\), the Hamiltonian is \(H\), and
\[
\tau=\frac{e^{-H/T}}{\Tr(e^{-H/T})}
\]
is the Gibbs state. Suppose first that \(\rho\) is rank deficient. Write
\[
\rho=\sum_{i=1}^r r_i\ket{i}\!\bra{i},
\qquad
r_i>0,
\qquad
r=\operatorname{rank}(\rho)<d,
\]
and let
\[
\Pi_0:=\id-\sum_{i=1}^r\ket{i}\!\bra{i},
\qquad
k:=d-r.
\]
The obstruction for rank-deficient \(\rho\) is that no finite Hamiltonian has \(\rho\) as an exact Gibbs state. 
The cutoff Hamiltonian \(G_\Lambda\) makes the kernel of \(\rho\) very costly but finite. 
For a cutoff \(\Lambda>0\), define
\[
G_\Lambda:=\sum_{i=1}^r(-T\log r_i)\ket{i}\!\bra{i}+\Lambda\Pi_0 .
\]
Its Gibbs state is
\[
\gamma_\Lambda
=
\frac{e^{-G_\Lambda/T}}{\Tr(e^{-G_\Lambda/T})}
=
\frac{\rho+e^{-\Lambda/T}\Pi_0}{1+k e^{-\Lambda/T}}.
\]
Thus \(\gamma_\Lambda\) is full rank for every finite \(\Lambda\), and \(\gamma_\Lambda\to\rho\) as \(\Lambda\to\infty\).

Consider the following protocol.
\begin{enumerate}
\item \emph{Sudden quench.} Isolate the system and quench the Hamiltonian from \(H\) to \(G_\Lambda\). The state remains \(\rho\). The extracted work is
\[
W_1=\Tr(\rho H)-\Tr(\rho G_\Lambda).
\]
\item \emph{Thermalization at fixed \(G_\Lambda\).} Couple the system to the bath and let it relax from \(\rho\) to \(\gamma_\Lambda\) \cite{BreuerPetruccione2002,KurizkiKofman2021}. Since the Hamiltonian is fixed, no work is extracted in this step.
\item \emph{Reversible isothermal stroke.} Starting from \(\gamma_\Lambda\), perform a quasistatic isothermal transformation from \(G_\Lambda\) to \(H\). Let \(K_s\) be the instantaneous Hamiltonian, with \(K_0=G_\Lambda\) and \(K_1=H\), and let
\[
\omega_s=\frac{e^{-K_s/T}}{Z_s},
\qquad
Z_s=\Tr(e^{-K_s/T}).
\]
The state remains \(\omega_s\) throughout the quasistatic stroke.
\end{enumerate}

Along the reversible isothermal stroke,
\[
\mathcal F_T(\omega_s,K_s)=-T\log Z_s,
\]
and differentiation gives
\[
\frac{d}{ds}\mathcal F_T(\omega_s,K_s)=\Tr\!\left(\omega_s\frac{dK_s}{ds}\right).
\]
With extracted work positive, the work extracted during the stroke is
\[
W_2=-\int_0^1 ds\,\Tr\!\left(\omega_s\frac{dK_s}{ds}\right)
=
\mathcal F_T(\gamma_\Lambda,G_\Lambda)-\mathcal F_T(\tau,H).
\]
The total extracted work is
\[
W_\Lambda=\Tr(\rho H)-\Tr(\rho G_\Lambda)+\mathcal F_T(\gamma_\Lambda,G_\Lambda)-\mathcal F_T(\tau,H).
\]
Equivalently,
\[
W_\Lambda=
\mathcal F_T(\rho,H)-\mathcal F_T(\tau,H)
-
\Bigl[\mathcal F_T(\rho,G_\Lambda)-\mathcal F_T(\gamma_\Lambda,G_\Lambda)\Bigr].
\]
The only loss relative to the reversible free-energy value is the mismatch between \(\rho\) and the finite-temperature Gibbs state \(\gamma_\Lambda\) of the cutoff Hamiltonian. 
Since \(\gamma_\Lambda\) is the Gibbs state of \(G_\Lambda\),
\[
\mathcal F_T(\rho,G_\Lambda)-\mathcal F_T(\gamma_\Lambda,G_\Lambda)=T D(\rho\|\gamma_\Lambda).
\]
On the support of \(\rho\), the eigenvalues of \(\gamma_\Lambda\) are \(r_i/(1+k e^{-\Lambda/T})\), and hence
\[
D(\rho\|\gamma_\Lambda)=\log(1+k e^{-\Lambda/T}).
\]
Therefore
\[
W_\Lambda=
\mathcal F_T(\rho,H)-\mathcal F_T(\tau,H)-T\log(1+k e^{-\Lambda/T}),
\]
and
\[
W_\Lambda\xrightarrow[\Lambda\to\infty]{}
\mathcal F_T(\rho,H)-\mathcal F_T(\tau,H).
\]
For finite \(\Lambda\), the inequality is strict whenever \(k>0\). If \(\rho\) is full rank, one may take \(G=-T\log\rho\) up to an additive constant, and the same construction saturates the free-energy bound without the cutoff.

\begin{remark}
For rank-deficient states, exact saturation requires the limiting Hamiltonian \(G_\Lambda\) to develop diverging gaps on the kernel of \(\rho\). This is why the finite-\(\Lambda\) protocol approaches, but does not exactly reach, the free-energy bound.
\end{remark}

\section{Proof of the classical average-work law}
\label{app:main-bound-proof}

We prove Theorem~\ref{thm:commuting-average-bound}. 
The proof has two parts. 
First, we show that the source-side energy constraints force the three pure preparations used in the minimal task to have a two-dimensional representation. 
This is not a qubit assumption on the work device: the device Hilbert space and the commuting Hamiltonian implementation remain arbitrary-dimensional. 
Only the source geometry relevant to the three supplied preparations collapses to the two-dimensional reference span. 
Second, commutativity of \(H_0\) and \(H_1\) permits a simultaneous decomposition of the two Hamiltonian settings. 
At each point of this decomposition, the reference pair induces a Bloch-disk coordinate. 
A rotation of this disk isolates the component fixed by the \(H_0\)-values, and Cauchy--Schwarz bounds how large the \(H_1\)-value on \(\ket{\psi_1}\) can be.

We first derive the source geometry from the source-side constraints. 
For pure preparations,
\[
\eta_{ij}
=
\frac12\lambda_{\max}(\rho_i+\rho_j)
=
\frac{1+|\langle\psi_i|\psi_j\rangle|}{2}.
\]
The bound \(\eta_{0\perp}\le1/2\) therefore implies
\[
|\langle\psi_0|\psi_\perp\rangle|=0,
\]
so \(\ket{\psi_0}\) and \(\ket{\psi_\perp}\) are orthogonal. 
Let
\[
\Pi_{\rm ref}:=\ket{\psi_0}\!\bra{\psi_0}+\ket{\psi_\perp}\!\bra{\psi_\perp}
\]
be the projector onto their span. 
Write
\begin{align*}
a&:=\langle\psi_0|\psi_1\rangle,\\
b&:=\langle\psi_\perp|\psi_1\rangle,\\
\ket{r}&:=(\id-\Pi_{\rm ref})\ket{\psi_1}.
\end{align*}
The two lower source constraints give
\[
|a|\ge \cos\theta,
\qquad
|b|\ge \sin\theta .
\]
Since \(0<\theta<\pi/2\), both \(\cos\theta\) and \(\sin\theta\) are positive. 
By normalization of \(\ket{\psi_1}\),
\[
1
=
|a|^2+|b|^2+\|r\|^2
\ge
\cos^2\theta+\sin^2\theta+\|r\|^2
=
1+\|r\|^2 .
\]
Hence \(\|r\|=0\), and equality must hold throughout:
\[
|a|=\cos\theta,
\qquad
|b|=\sin\theta .
\]
Thus the source-side inequalities are saturated by normalization. 
After choosing phases of \(\ket{\psi_1}\) and \(\ket{\psi_\perp}\), we may write
\[
\ket{\psi_1}
=
\cos\theta\,\ket{\psi_0}
+
\sin\theta\,\ket{\psi_\perp}.
\]
This is the only reduction to a two-dimensional structure in the proof, and it is forced by the thermodynamic source constraints. 
No restriction is placed on the dimension of the work device or on the dimension of the commuting Hamiltonian implementation.

Set \(c=\cos\theta\), \(s=\sin\theta\), \(C=\cos2\theta\), and \(S=\sin2\theta\). 
For notational clarity, we write the proof for a discrete simultaneous decomposition of the two commuting settings. 
For general bounded commuting Hamiltonians, the same argument uses the standard measure representation for commuting operators, with sums replaced by integrals. 
If \(H_0,H_1\) commute, then in the discrete case they may be written as
\begin{align*}
H_0&=\sum_k\alpha_k\ket{k}\!\bra{k},
&
0&\le \alpha_k\le1,\\
H_1&=\sum_k\beta_k\ket{k}\!\bra{k},
&
0&\le \beta_k\le1 .
\end{align*}
Write
\[
\ket{\psi_0}=\sum_k x_k\ket{k},
\qquad
\ket{\psi_\perp}=\sum_k y_k\ket{k}.
\]
For every \(k\) with \(|x_k|^2+|y_k|^2>0\), define
\begin{align*}
w_k&:=\frac{|x_k|^2+|y_k|^2}{2},\\
u_k&:=\frac{|x_k|^2-|y_k|^2}{|x_k|^2+|y_k|^2},
&
v_k&:=\frac{2\operatorname{Re}(\overline{x_k}y_k)}{|x_k|^2+|y_k|^2}.
\end{align*}
When the denominator vanishes, set \(w_k=u_k=v_k=0\). 
The number \(w_k\) is the average probability weight of the reference pair at the decomposition point \(k\). 
The coordinate \(u_k\) records the imbalance between \(\ket{\psi_0}\) and \(\ket{\psi_\perp}\), while \(v_k\) records their real coherence in the simultaneous decomposition. 
The pair \((u_k,v_k)\) lies in the unit disk, since
\[
u_k^2+v_k^2\le1.
\]
Moreover,
\[
\sum_k w_k=1,
\qquad
\sum_k w_k u_k=0,
\qquad
\sum_k w_k v_k=0.
\]
The last two equalities follow from the normalization and orthogonality of the reference pair.

Rotate the disk coordinates by
\[
t_k:=Cu_k+Sv_k,
\qquad
z_k:=-Su_k+Cv_k.
\]
Then
\[
t_k^2+z_k^2\le1,
\qquad
\sum_k w_kt_k=0,
\qquad
\sum_k w_kz_k=0.
\]
Using \(\ket{\psi_1}=c\ket{\psi_0}+s\ket{\psi_\perp}\), one obtains
\[
|\langle k|\psi_1\rangle|^2
=
w_k(1+t_k).
\]

Now define
\[
A_k:=\alpha_k-\mu_0,
\qquad
B_k:=\beta_k-\mu_1,
\qquad
\mu_j:=\frac{a_{0,j}+a_{\perp,j}}{2}.
\]
Since
\[
\mu_0=\sum_k w_k\alpha_k,
\qquad
\mu_1=\sum_k w_k\beta_k,
\]
we have \(\sum_k w_kA_k=\sum_k w_kB_k=0\). 
The value of \(H_0\) on \(\ket{\psi_1}\) gives
\[
a_{1,0}-\mu_0
=
\sum_k w_kA_kt_k.
\]
The value of \(H_0\) on \(\ket{\psi_0}\) gives
\[
a_{0,0}-\mu_0
=
\sum_k w_kA_ku_k
=
C\sum_k w_kA_kt_k
-
S\sum_k w_kA_kz_k .
\]
Thus, with
\[
Z_\theta:=\sum_k w_kA_kz_k,
\]
one obtains
\[
Z_\theta
=
\frac{
\cos(2\theta)(a_{1,0}-\mu_0)-(a_{0,0}-\mu_0)
}{
\sin(2\theta)
},
\]
which is exactly the expression stated in Theorem~\ref{thm:commuting-average-bound}.

To see the first Cauchy--Schwarz step explicitly, use the weighted inner product
\[
\langle f,g\rangle_w:=\sum_k w_k f_kg_k .
\]
Then \(a_{1,0}-\mu_0=\langle A,t\rangle_w\) and \(Z_\theta=\langle A,z\rangle_w\). 
For any real numbers \(r,s\) with \(r^2+s^2=1\),
\[
r(a_{1,0}-\mu_0)+sZ_\theta
=
\langle A,rt+sz\rangle_w .
\]
By Cauchy--Schwarz,
\[
\langle A,rt+sz\rangle_w^2
\le
\langle A,A\rangle_w\,\langle rt+sz,rt+sz\rangle_w .
\]
Pointwise, \((rt_k+sz_k)^2\le t_k^2+z_k^2\le1\). 
Taking the supremum over \(r,s\) with \(r^2+s^2=1\) gives
\[
(a_{1,0}-\mu_0)^2+Z_\theta^2
\le
\sum_k w_kA_k^2 .
\]
Because \(0\le\alpha_k\le1\) and \(\sum_k w_k\alpha_k=\mu_0\),
\[
\sum_k w_kA_k^2
=
\sum_k w_k(\alpha_k-\mu_0)^2
\le
\mu_0(1-\mu_0)
=
R_0^2 .
\]
Therefore
\[
(a_{1,0}-\mu_0)^2+Z_\theta^2\le R_0^2.
\]
In particular, \(Z_\theta^2\le R_0^2\) when \(0<\mu_0<1\).

The second Cauchy--Schwarz inequality bounds the desired value of \(H_1\) on \(\ket{\psi_1}\). 
Since
\[
a_{1,1}-\mu_1
=
\sum_k w_kB_kt_k,
\]
we have
\[
(a_{1,1}-\mu_1)^2
\le
\left(\sum_k w_kB_k^2\right)
\left(\sum_k w_kt_k^2\right).
\]
Again using \(0\le\beta_k\le1\) and \(\sum_k w_k\beta_k=\mu_1\),
\[
\sum_k w_kB_k^2
\le
\mu_1(1-\mu_1)
=
R_1^2 .
\]
It remains to bound the available \(t\)-variance. 
From the definition of \(Z_\theta\),
\[
Z_\theta^2
\le
\left(\sum_k w_kA_k^2\right)
\left(\sum_k w_kz_k^2\right)
\le
R_0^2\sum_k w_kz_k^2 .
\]
Hence
\[
\sum_k w_kz_k^2
\ge
\frac{Z_\theta^2}{R_0^2}.
\]
Since \(t_k^2+z_k^2\le1\),
\[
\sum_k w_kt_k^2
\le
1-\sum_k w_kz_k^2
\le
1-\frac{Z_\theta^2}{R_0^2}.
\]
Combining the last three inequalities gives
\[
a_{1,1}
\le
\mu_1
+
R_1\sqrt{1-\frac{Z_\theta^2}{R_0^2}}.
\]
Adding \(a_{0,0}\) and dividing by two proves Eq.~\eqref{eq:commuting-average-bound}. 
If \(\mu_0\in\{0,1\}\), then \(R_0=0\), so the previous division by \(R_0\) is not used. 
In this endpoint case, boundedness forces \(a_{0,0}=a_{\perp,0}=\mu_0\), and the same argument with only \(\sum_k w_kt_k^2\le1\) gives the endpoint form stated in Theorem~\ref{thm:commuting-average-bound}.

\section{Proof of the one-setting envelope}
\label{app:one-setting-envelope}

This lemma bounds what a single unrestricted bounded Hamiltonian can do at a fixed reference average. 
Let
\[
P=\ket{\psi_0}\!\bra{\psi_0}+\ket{\psi_\perp}\!\bra{\psi_\perp}
\]
be the projector onto the two-dimensional reference span, and let \(B=PHP\) be the compression of \(H\) to this span. 
Since \(0\le H\le\id\), the compressed operator satisfies
\[
0\le B\le P .
\]
Equivalently, inside the reference span, \(B\) is a \(2\times2\) positive operator bounded above by the identity on that span.

The reference-average condition is
\[
\mu=\frac{\langle\psi_0|H|\psi_0\rangle+\langle\psi_\perp|H|\psi_\perp\rangle}{2}.
\]
Since \(\{\ket{\psi_0},\ket{\psi_\perp}\}\) is an orthonormal basis of the reference span, the trace of \(B\) on this span is
\begin{align*}
\Tr B
&=
\langle\psi_0|B|\psi_0\rangle+\langle\psi_\perp|B|\psi_\perp\rangle\\
&=
\langle\psi_0|H|\psi_0\rangle+\langle\psi_\perp|H|\psi_\perp\rangle\\
&=
2\mu .
\end{align*}
Thus \(0\le\mu\le1\), and \(2\mu\) may range from \(0\) to \(2\).

For any unit vector \(\ket{\phi}\) in the reference span,
\[
\langle\phi|H|\phi\rangle
=
\langle\phi|B|\phi\rangle
\le
\lambda_{\max}(B).
\]
Because \(0\le B\le P\), every eigenvalue of \(B\) is at most \(1\). 
Because \(\Tr B=2\mu\), no eigenvalue of \(B\) can exceed \(2\mu\). 
Therefore
\[
\lambda_{\max}(B)\le \min\{1,2\mu\}.
\]

The upper bound \(\lambda_{\max}(B)\le\min\{1,2\mu\}\) is attainable inside the reference span. 
If \(2\mu\le1\), choose
\[
B=2\mu\,\ket{\phi}\!\bra{\phi}.
\]
If \(2\mu\ge1\), choose
\[
B=\ket{\phi}\!\bra{\phi}
+
(2\mu-1)\ket{\phi^\perp}\!\bra{\phi^\perp},
\]
where \(\ket{\phi^\perp}\) is the unit vector in the reference span orthogonal to \(\ket{\phi}\). 
In both cases \(0\le B\le P\), \(\Tr B=2\mu\), and the value on \(\ket{\phi}\) equals \(\min\{1,2\mu\}\). 
Extending \(B\) by zero outside the reference span gives an admissible bounded Hamiltonian on the full Hilbert space.

Hence
\[
\max_H\langle\phi|H|\phi\rangle
=
\min\{1,2\mu\}
=
\mu+\min\{\mu,1-\mu\}.
\]
This proves Lemma~\ref{lem:one-setting-envelope}.

\section{Optimality of the minimal-task advantage}
\label{app:optimal-minimal-advantage}

We prove Theorem~\ref{thm:optimal-minimal-average-advantage}. 
The proof first upper-bounds the advantage of any minimal-task instance and then analyzes equality conditions. 
This order is important: the symmetric source and balanced work scale are derived, not assumed. 
The optimization is over the full minimal-task family: the source angle \(0<\theta<\pi/2\), the reference averages \(\mu_0,\mu_1\), all bounded Hamiltonian implementations without a commutativity restriction, and all bounded commuting implementations.

For fixed reference averages, define
\[
\ell_j:=\min\{\mu_j,1-\mu_j\},
\qquad j=0,1.
\]
By Lemma~\ref{lem:one-setting-envelope}, any bounded implementation, commuting or not, satisfies
\[
a_{0,0}\le \mu_0+\ell_0,
\qquad
a_{1,1}\le \mu_1+\ell_1.
\]
Hence every incompatible implementation obeys
\begin{equation}
\label{eq:q-envelope-absolute}
\overline W_p^{\,\rm q}
\le
\frac12(\mu_0+\mu_1+\ell_0+\ell_1).
\end{equation}

To upper-bound the quantum advantage, it is enough to show that every task instance has at least one sufficiently good commuting implementation. 
The following construction uses one shared direction \(u\) for both settings. 
We now lower-bound the performance that a commuting device can always achieve at the same reference averages. Let \(m_0,m_1\) be the Bloch vectors of \(\psi_0,\psi_1\) in the reference span, so that
\[
m_0\cdot m_1=\cos2\theta.
\]
For any unit vector \(u\) and signs \(s_0,s_1\in\{\pm1\}\), define
\[
K_j^{\rm cl}
=
\mu_j\id+\ell_j s_j\,u\cdot\sigma,
\qquad j=0,1.
\]
These two bounded Hamiltonians commute, since they are both functions of \(u\cdot\sigma\). Their scored average work is
\[
\overline W_p^{\,\rm cl}
=
\frac12\left[
\mu_0+\mu_1
+
\ell_0s_0\,u\cdot m_0
+
\ell_1s_1\,u\cdot m_1
\right].
\]
Optimizing over \(u\) and over the signs gives
\begin{equation}
\label{eq:commuting-lower-absolute}
\overline W_p^{\,\rm cl}
\ge
\frac12\left[
\mu_0+\mu_1+
\sqrt{
\ell_0^2+\ell_1^2+2\ell_0\ell_1|\cos2\theta|
}
\right].
\end{equation}
Therefore every advantage satisfies
\begin{align}
\Delta_{\rm av}
&=
\overline W_p^{\,\rm q}
-
\overline W_p^{\,\rm cl}
\nonumber\\
&\le
\frac12\left[
\ell_0+\ell_1
-
\sqrt{
\ell_0^2+\ell_1^2+2\ell_0\ell_1|\cos2\theta|
}
\right].
\label{eq:gap-upper-exact-form}
\end{align}
This is an upper bound on the average-work advantage for every instance of the minimal task.

The remaining expression is now purely geometric: it asks how far two positive lengths \(\ell_0,\ell_1\) can be separated from one shared classical direction. 
For fixed \(\ell_0,\ell_1\), the right-hand side of Eq.~\eqref{eq:gap-upper-exact-form} is maximized by minimizing the square-root term. Hence any maximizer must satisfy
\[
|\cos2\theta|=0,
\qquad\text{so}\qquad
\theta=\frac\pi4.
\]
At this source angle,
\[
\Delta_{\rm av}
\le
\frac12\left[
\ell_0+\ell_1-\sqrt{\ell_0^2+\ell_1^2}
\right].
\]
The function
\[
f(\ell_0,\ell_1):=\ell_0+\ell_1-\sqrt{\ell_0^2+\ell_1^2}
\]
is increasing in each variable for \(\ell_0,\ell_1\ge0\). Since \(\ell_j\le1/2\), the maximum is attained only at
\[
\ell_0=\ell_1=\frac12,
\qquad\text{equivalently}\qquad
\mu_0=\mu_1=\frac12.
\]
Thus imbalance in the reference averages cannot help: it reduces the maximum possible value before it can increase the separation. 
Substituting these forced optimal values gives the global upper bound
\begin{equation}
\label{eq:absolute-minimal-gap-upper}
\Delta_{\rm av}
\le
\frac12\left(1-\frac1{\sqrt2}\right).
\end{equation}
Thus the symmetric source angle and the balanced reference averages are outputs of the optimization, not assumptions.

It remains to show that the bound is attainable. At the forced point
\[
\theta=\frac\pi4,
\qquad
\mu_0=\mu_1=\frac12,
\]
take the incompatible rank-one Hamiltonian settings
\[
H_0=\ket{\psi_0}\!\bra{\psi_0},
\qquad
H_1=\ket{\psi_1}\!\bra{\psi_1}.
\]
They give
\[
\overline W_p^{\,\rm q}=1.
\]
For commuting Hamiltonian settings, Theorem~\ref{thm:commuting-average-bound} gives the matching upper bound at this forced point. Indeed, with \(\theta=\pi/4\) and \(\mu_0=\mu_1=1/2\), write \(x:=a_{0,0}-1/2\). Then \(Z_\theta=-x\), and the classical law gives
\[
\overline W_p^{\,\rm cl}
\le
\frac12\left[
1+x+\frac12\sqrt{1-4x^2}
\right]
\le
\frac12\left(1+\frac1{\sqrt2}\right),
\]
where the last inequality is the maximum over \(-1/2\le x\le1/2\). The commuting construction in Eq.~\eqref{eq:commuting-lower-absolute} attains this value. Therefore
\[
\overline W_p^{\,\rm cl}
=
\frac12\left(1+\frac1{\sqrt2}\right).
\]
Consequently
\[
\Delta_{\rm av}^{\max}
=
1-
\frac12\left(1+\frac1{\sqrt2}\right)
=
\frac12\left(1-\frac1{\sqrt2}\right).
\]
For pure preparations, the forced source point corresponds to
\[
\eta_{0\perp}=\frac12,
\qquad
\eta_{01}=\eta_{\perp1}
=
\frac12\left(1+\frac1{\sqrt2}\right).
\]
This proves Theorem~\ref{thm:optimal-minimal-average-advantage}.

\section{Proof of the hierarchy threshold}
\label{app:hierarchy-threshold}

We prove Theorem~\ref{thm:hierarchy-threshold}. 
The hierarchy is a qubit pure-source corollary of the thermodynamic framework. 
The source family is specified by pairwise maximal-energy data corresponding to pure qubit preparations with Bloch-sphere directions \(M_N=\{m_x\}_{x=1}^N\). 
For pure qubit preparations,
\[
\eta_{xx'}
=
\frac12\left(1+\sqrt{\frac{1+m_x\cdot m_{x'}}{2}}\right),
\]
so the pairwise maximal-energy data are thermodynamic data for the same source geometry. 
The relevant source span may be represented as a qubit with
\[
\psi_x=\frac12(\id_2+m_x\cdot\sigma),
\qquad \|m_x\|=1 .
\]
This is a restriction on the source family, not on the work device. 
The work device and its Hamiltonian settings remain arbitrary-dimensional. 
Only the compression of each Hamiltonian setting to the source span enters the task average.

The balanced hierarchy scores the uniformly weighted branches \((x,x)\), so
\[
\overline W_p=\frac1N\sum_{x=1}^N \Tr(\psi_x H_x).
\]
Let \(P\) be the projector onto the source qubit span and define the compressed Hamiltonian
\[
K_x:=PH_xP .
\]
Since \(0\le H_x\le\id\), one has \(0\le K_x\le P\). 
The balanced reference condition fixes the average of every compressed setting on the reference pair to \(1/2\), equivalently \(\Tr K_x=1\) on the two-dimensional source span. 
Identifying \(P\) with \(\id_2\) on this span, every such compressed Hamiltonian has the Bloch form
\[
K_x=\frac12(\id_2+a_x\cdot\sigma),
\qquad
\|a_x\|\le1 .
\]
The score becomes
\[
\overline W_p=\frac12+\frac{1}{2N}\sum_{x=1}^N m_x\cdot a_x .
\]

We first upper-bound every commuting implementation. Since the full Hamiltonian settings commute, use the same simultaneous decomposition as in Appendix~\ref{app:main-bound-proof}. The induced source statistics can be described by weights \(w_k\), source Bloch coordinates \(r_k\) with \(\|r_k\|\le1\), and Hamiltonian values \(h_x(k)\in[0,1]\). The reference decomposition satisfies \(\sum_k w_k r_k=0\). Balancedness gives \(\sum_k w_k h_x(k)=1/2\) for every \(x\). Define \(A_x(k):=2h_x(k)-1\), so \(|A_x(k)|\le1\). The excess over \(1/2\) is
\[
\overline W_p-\frac12
=
\frac{1}{2N}\sum_k w_k\, r_k\cdot\sum_{x=1}^N A_x(k)m_x .
\]
For every \(k\),
\begin{align*}
r_k\cdot\sum_x A_x(k)m_x
&\le
\left\|\sum_x A_x(k)m_x\right\| \\
&\le
\max_{s_x\in\{\pm1\}}\left\|\sum_x s_xm_x\right\|.
\end{align*}
For fixed \(r_k\), the best choice of \(h_x(k)\) follows the sign of each projection. 
Equivalently, the norm is convex in each coefficient \(A_x\in[-1,1]\), so the maximum is attained at an extreme point. Hence every commuting implementation satisfies
\[
\overline W_p
\le
\frac12+\frac{1}{2N}\max_{s_x=\pm1}\left\|\sum_{x=1}^N s_xm_x\right\|.
\]
By duality of the Euclidean norm,
\[
\max_{s_x=\pm1}\left\|\sum_x s_xm_x\right\|
=
\max_{\|u\|=1}\sum_x |u\cdot m_x|.
\]
Therefore
\[
\overline W_p^{\rm cl}
\le
\frac12\left(1+R(M_N)\right),
\]
where
\[
R(M_N):=\max_{\|u\|=1}\frac1N\sum_x |u\cdot m_x|.
\]

The upper bound is not merely a relaxation. It is attained by a two-value commuting device whose shared direction is the best alignment direction \(u_*\). Let \(u_*\) attain the maximum in \(R(M_N)\), and choose signs \(s_x=\operatorname{sgn}(u_*\cdot m_x)\), with arbitrary signs when the scalar product vanishes. Use the two-value commuting Hamiltonian decomposition with \(F_\pm=(\id\pm u_*\cdot\sigma)/2\), dilated if necessary. Equivalently, use a two-point commuting implementation with decomposition points \(k=\pm\), weights \(w_+=w_-=1/2\), and source Bloch coordinates \(r_+=u_*\), \(r_-=-u_*\). Define
\[
h_x(+)=\frac{1+s_x}{2},
\qquad
h_x(-)=\frac{1-s_x}{2}.
\]
Then each setting is bounded, all settings commute, and each has reference average \(1/2\). The achieved score is
\[
\overline W_p^{\rm cl}
=
\frac12+\frac{1}{2N}\sum_x |u_*\cdot m_x|
=
\frac12\left(1+R(M_N)\right).
\]
Thus the commuting benchmark is exact.

The noisy incompatible implementation uses
\[
H_x^{(v)}=v\psi_x+(1-v)\frac{\id_2}{2}.
\]
For the scored branch \((x,x)\),
\[
\Tr(\psi_x H_x^{(v)})=\frac12(1+v).
\]
Therefore
\[
\overline W_p^{\rm q}(v)=\frac12(1+v).
\]
The implementation is advantageous precisely when
\[
\frac12(1+v)>\frac12(1+R(M_N)),
\]
or equivalently
\[
v>R(M_N).
\]
Thus the visibility threshold is \(v_c(M_N)=R(M_N)\).

For evenly spaced equatorial directions modulo antipodes, the maximizing direction may be taken halfway between two neighboring directions or aligned with one direction, depending on parity. Summing the resulting absolute cosines gives the following identity. For directions at angles
\[
\phi_x=\frac{(x-1)\pi}{N},
\qquad x=1,\ldots,N,
\]
one has the finite trigonometric identity
\[
\max_\phi\sum_{x=1}^N |\cos(\phi-\phi_x)|
=
\csc\!\left(\frac{\pi}{2N}\right).
\]
The maximizer is aligned with one source direction for odd \(N\), and halfway between two adjacent directions for even \(N\). Hence
\[
v_c(N)=\frac1N\csc\!\left(\frac{\pi}{2N}\right)
\longrightarrow \frac2\pi .
\]
For directions approaching the full Bloch-sphere limit,
\[
R=\int_{S^2}\frac{d\Omega}{4\pi}|u\cdot m|
=\frac12,
\]
so the limiting threshold is \(v_c=1/2\). This proves Theorem~\ref{thm:hierarchy-threshold}.

\end{document}